\documentclass[12pt]{amsart}
\usepackage{graphicx}

\newtheorem{theorem}{Theorem}[section]

\newtheorem{proposition}[theorem]{Proposition}

\newtheorem{corollary}[theorem]{Corollary}

\theoremstyle{definition}
\newtheorem{definition}[theorem]{Definition}
\newtheorem{remark}[theorem]{Remark}
\newtheorem{example}[theorem]{Example}

\def\cal{\mathcal}

\title[Analysis of inertial waves]{Mathematical analysis of inertial waves
in rectangular basins with one sloping boundary}

\author{S.\,D. Troitskaya}
\email{troitsks@gmail.com}

\address{Laboratory of Mathematical Education, Institute on Educational Content and Methods
of the Russian Academy of Education, Moscow,  105062 Russia\\
Moscow State University, Moscow,  119991 Russia}

\begin{document}

\begin{abstract}
Here we consider the problem of small oscillations of a rotating inviscid
incompressible fluid.

From a mathematical point of view, new exact solutions to the two-dimensional
Poincar\'e-Sobolev equation in a class of domains including trapezoid are found
in an explicit  form and their main properties are described.
These solutions correspond to the absolutely continuous spectrum of a linear operator
that is associated with this system of equations.

For specialists in Astrophysics and Geophysics the existence of these solutions
signifies the existence of some previously unknown type of inertial waves corresponding
to the continuous spectrum of inertial oscillations. A fundamental distinction
between monochromatic inertial waves and  waves of the new type is shown:
usual characteristics (frequency, amplitude, wave vector,  dispersion relation,
direction of energy propagation, and so on) are not applicable to the last.
Main properties of these waves are described. In particular it is proved
that they are progressive.  Main features of their energy transfer are described.
The existence of such inertial waves enables us to explain in a new way
a lot of experimental data that were obtained in Geophysics in the past two decades
and to predict  the occurrence of such oscillations in natural waters.
\end{abstract}

%Rotating fluids}

\maketitle

\section{Introduction}
The paper was inspired by numerous articles in Geophysics and Astrophysics
of the past two decades  which are devoted to the  phenomenon  of the kinetic energy localization of internal waves
in rotating fluids, in stratified fluids, in rotating stratified fluids,
i.e. of the phenomenon of emergence of so-called \textit{wave attractors}
(see, e.g., \cite{MaasBenielliNature97,RieutordValdet00,StaquetSommeria2002,HarlanMaas07,%
SwartSleijpenMaasBrandts2007,Harland08p1,GerkemaZimmerman2008,GrisouardStaquetPairaud2008,GerkemaMaas2013,%
ScolanErmDauxois2013}).
Here we will describe briefly the situation by the example of rotating fluids.

 The motion of an ideal incompressible fluid contained in a region $G$, which
 rotates with constant velocity about the axis $\mathbf{k}$, is governed by the equations:
\begin{equation}\label{0gl.1p.04aA}
\mathbf{U}_{t}+\varepsilon\mathbf{U}\cdot \nabla \mathbf{U}+
2\mathbf{k}\times \mathbf{U}= - \nabla p
\qquad  \mbox{in }G,
\end{equation}
$$
%\begin{equation}\label{0gl.1p.05a}
\nabla\cdot \mathbf{U}=0 \qquad  \mbox{in }G,
%\end{equation}
$$
$$
%\begin{equation}\label{0gl.1p.06a}
\bigl.\bigl( \mathbf{U}\cdot \mathbf{n}\bigr)\bigr|_{\partial G}=0.
%\end{equation}
$$
Here $\mathbf{U}=(u,\,v,\,w)$ is the velocity field in the rotating coordinate system
rigidly connected to the container $G$, $p$ is the hydrodynamic pressure,  $\mathbf{n}$
is the unit outward normal to the boundary $\partial G$ and $\varepsilon$ is constant (Rossby number)
(see \cite{GreenspanBook68}).
The linearization of (\ref{0gl.1p.04aA})
 in a neighborhood of the solution corresponding to the rotation of the fluid as
a rigid body reduces the previous system  to
\begin{equation}\label{0gl.1p.04a}
\mathbf{U}_{t}+
2\mathbf{k}\times \mathbf{U}= - \nabla p
\qquad  \mbox{in }G,
\end{equation}
\begin{equation}\label{0gl.1p.05a}
\nabla\cdot \mathbf{U}=0 \qquad  \mbox{in }G,
\end{equation}
\begin{equation}\label{0gl.1p.06a}
\bigl.\bigl( \mathbf{U}\cdot \mathbf{n}\bigr)\bigr|_{\partial G}=0.
\end{equation}
Solutions $\mathbf{U}$ to this system are called \emph{inertial waves}.
If a solution has the form $e^{-i\mu t}\mathbf{U}(x,y,z)$,
then it is called  an \emph{inertial mode} with the natural frequency $\mu $.
The study of this problem was initiated by H. Poincar\'e
in \cite{Poin85}, and then continued by many authors.

Below without lost of generality we may take
$\mathbf{k}=(0,0, \frac 12)$.

It is known that properties of solutions of
(\ref{0gl.1p.04a}--\ref{0gl.1p.06a}) strongly
depend on the shape of the container $G$.
According to numerous papers, one of containers where
the phenomenon of the energy localization of the inertial waves can take place,
is a cylindrical tank which is highly elongated in the direction $Oy$
and is such that its linear dimensions are small in comparison with its distance from the axis of rotation.
So it is reasonable to assume that $G$ is an infinitely long
 cylinder: $G = \{(x, y, z)|(x, z) \in D, y \in \mathbf{R}\}$,
and that the components of the velocity $ \mathbf{U}=(u,\,v,\,w)$ and the pressure $p$ depend
only  on time $t$ and two spatial variables $x$ and $z$
(see \cite{Barcilon68}\footnote{
It was shown in \cite{Masl68eng} that the study of this two-dimensional problem
is very important for predicting properties of solutions
of the three-dimensional problem (\ref {0gl.1p.04a}--\ref {0gl.1p.06a})
in cylindrical domains with finite lengths.}).
Therefore, the considered initial-boundary value problem has the form:
\begin{equation}\label{0gl.1p.4}
\frac {\partial u}{\partial t}= v - \frac {\partial p}{\partial x},\,\,
\frac {\partial v}{\partial t}= - u, \,\, \frac {\partial w}{\partial t}
= - \frac {\partial p}{\partial z}, \quad
(x, z) \in D,
\end{equation}
\begin{equation}\label{0gl.1p.6}
\frac {\partial u}{\partial x}+\frac {\partial w}{\partial z}=0 \qquad\mbox{in }D,
\end{equation}
\begin{equation}\label{0gl.1p.10}
\bigl.\bigl( u n_1+w n_3\bigr)\bigr|_{\partial D}=0,
\end{equation}
\begin{equation}\label{init1}
u|_{t=0}=u_0,\quad v|_{t=0}=v_0,\quad
w|_{t=0}=w_0,
\end{equation}
where
$\mathbf{n}=(n_1,\,n_3)$ is the unit outward normal to ${\partial D}$ in the plane $Oxz$.

It is known that in this case the stream function $\psi$
corresponding to (\ref{0gl.1p.4}--\ref{init1}) exists and $\psi$
is a solution to the following initial boundary value problem:
\begin{equation}\label{0gl.1p.16}
\frac {\partial^2 }{\partial t^2}\left( \frac {\partial^2 \psi}{\partial x^2}+
\frac {\partial^2 \psi}{\partial z^2} \right)+
\frac {\partial^2 \psi}{\partial z^2}=0,\footnote{The equation (\ref{0gl.1p.16}) is often called the
(two-dimensional) Poincar\'e-Sobolev equation.}\quad
(x,z;\,t)\in D\times (0,\infty),
\end{equation}
\begin{equation}\label{0gl.1p.18}
\psi|_{\partial D\times (0,\infty)}=0,
\end{equation}
\begin{equation}\label{05gl.1p.20}
\psi|_{t=0}=\psi_0,\quad \psi_t|_{t=0}=\psi_1,\quad
(\psi_0|_{\partial D}=0, \psi_1|_{\partial D}=0).
\end{equation}
Papers of many authors are devoted to the study of this problem, however,
all the studies are far from being complete.
Until now exact solutions in an explicit form have been
found  for a very small number of configurations. One of the reasons for this situation
is that the problem of finding  inertial modes for this system
leads to the problem of the existence of nontrivial solutions to the Dirichlet problem for
 the string oscillation equation which is known for its complexity:
\begin{equation}\label{10}
\frac {\partial^2 \psi}{\partial x^2}-\frac 1{a^2} \frac {\partial^2 \psi}{\partial
z^2}=0, \quad (x,z) \in D, \quad a^2=\frac{\lambda}{1-\lambda}, \quad \lambda \in (0,1),
\end{equation}
\begin{equation}\label{06gl.1p.18}
\psi|_{\partial D}=0.
\end{equation}

It is known that for an arbitrary solution of (\ref{0gl.1p.04a}--\ref{0gl.1p.06a}),
the law of conservation of kinetic energy holds (see, for example,
\cite{KopachevskyKrein2001engV1}).
For solutions of (\ref{0gl.1p.4}--\ref{init1}) it takes the following form:
\begin{equation}\label{energyLaw2dim}
{\cal E}(t,D):=\int\limits_D \left(\left|u\right|^2+
\left|v\right|^2+\left|w\right|^2\right)  dxdz={\rm const}.
\end{equation}

The problem that was considered in the cited papers, is that for some domains $D$,
there exist some points or lines (wave attractors)  in $\overline{D}=D \cup \partial D$,
such that in the course of time a certain quantity $K>0$ of the kinetic energy
of the inertial waves concentrates within arbitrarily small neighborhoods of these points (or lines).
In general this phenomenon has a destructive nature.
An analogous problem arises in stratified fluids and in stratified rotating fluids,
and it is not surprising because it is well-known that the equations
that describe internal waves in such fluids, in many aspects are similar to the above equations
 for inertial waves (see, for example, \cite{Lamb32,GabovSveshn1986}).

If $D$ is a rectangle or an ellipse whose axes of symmetry are parallel to the coordinate axes then
(\ref{0gl.1p.4}--\ref{init1}) possess a full system of eigenmodes and
all inertial waves are almost periodic functions in time
\cite{Alexandryan1949eng,Alexandryan1960eng,Denchev1eng59}.
Thus wave attractors can not occur in these domains.

For the first time the existence of
some special types of oscillations was noted in geophysical
papers \cite{Green69,Wunsch1968,Wunsch1969}.
In particular, in \cite{Wunsch1968,Wunsch1969} the author described an experiment  that  proved
the possibility of the energy localization of the stratified fluid in a cylindrical tank having
the described above configuration with some right triangle as a cross-section $D$.
In \cite{TroitRGMF10,TroitBull_const_10,TroitBull_prop_10}
the existence of such oscillations in a rotating fluid has been
mathematically proved:
namely, for such right triangles  some class of exact solutions to the problem
(\ref{0gl.1p.4}--\ref{0gl.1p.10}) was found in an explicit form, and it was proved
that these oscillations are such that in the course of time, the energy of the initial state
of the fluid turns out to be almost completely concentrated in arbitrarily small neighborhoods of
the vertices  of the triangle. Thus, according to  the current terminology
in Geophysics and Astrophysics, these vertices may be called \textit{point wave attractors}.

Many theoretical and experimental geophysical
papers appeared in the past two decades which claim that an area where
a wave attractor may occur in a rotating fluid or in a stably stratified fluid,
is a cylindrical tank with the trapezoidal cross-section of the form
\begin{equation}\label{trapezoidD}
D=\{(x,z):-1<z<0,\quad -1<x<z+1\}.
\end{equation}
This tank is usually called a``rectangular basin with one sloping boundary''. Here
the concrete linear dimensions of $D$ are not important, but it is important that
$D$ is a non-isosceles trapezoid whose base is parallel to the axis $Ox$.

Below (unless otherwise specified) we consider namely this domain $D$.

It is easy to establish that for each
 $\lambda \in (\frac 12,\frac 45)$
 there exists a unique parallelogram $P(\lambda)$,
inscribed in $D$ in such a way that its sides are parallel
to characteristics of (\ref{10}) (see \textbf{Figure \ref{fig:ma1}}a).).
Denote by $S(\lambda)$ the boundary of $P(\lambda)$.
In \cite{MandersMaas2003} it was obtained experimentally, that
under certain perturbations of the uniform rotation, the motion of the fluid particles
in the rotating coordinate system
is intensive only in small neighborhoods of $S(\lambda)$
(where $S(\lambda)$ corresponds to some value $\lambda$ of the considered interval)
and that in other areas of the trapezoid the motion of particles is practically absent
(see Figure \ref{fig:ma1}a).).
An analogous result for stratified fluids was obtained in \cite{MaasBenielliNature97}.
\begin{figure}[ht]
\begin{center}
\includegraphics[height=3.7cm]{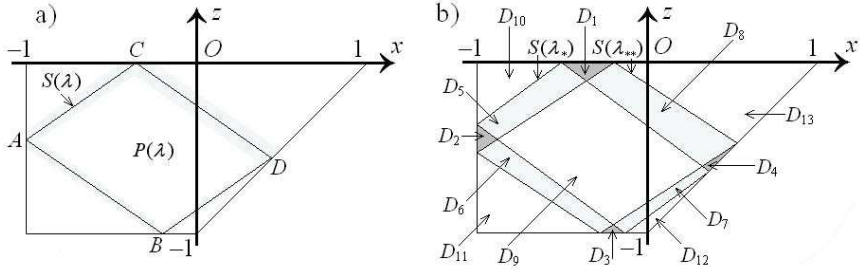}
\caption{a). For $\lambda \in (\frac 12,\frac 45)$,
a unique parallelogram $P(\lambda)$ can be constructed such that it is
inscribed in $D$ and its sides are parallel to the straight lines of the type
$x+az=C_1$, $x-az=C_2$, where $a=\sqrt{\frac{\lambda}{1-\lambda}}$.
\newline
b). For $1/2<\lambda_*<\lambda_{**}<4/5 $,
 $S(\lambda_ *) $ and $S(\lambda_{* *}) $ divide $D$
into subdomains $ D_i $, $i=1, 2,..,13$.
}
\label{fig:ma1}
\end{center}
\end{figure}

These facts as well as a preliminary study of the problem by the
%so-called
``rays method''  led the researchers to assume  that there are solutions of the problem
(\ref {0gl.1p.16}--\ref {05gl.1p.20}) localizing their energy
in an arbitrarily small neighborhood of $ S(\lambda) $
in process of time; i.e., that $ S(\lambda) $ is a wave attractor
(see, e.g., \cite{MaasBenielliNature97,%
MaasFlMech2001,%
StaquetSommeria2002,%
MandersMaas2003,%
MandersMaas2004,%
MaasChaos2005,%
HarlanMaas07,%
HazewinkelMaasDalziel2007,%
SwartSleijpenMaasBrandts2007,%
GerkemaZimmerman2008,%
GerkZimmMaasHaren2008,%
GerkemaMaas2013,%
GrisouardStaquetPairaud2008,%
Harland08p1,%
LamMaas2008,%
HazewinkelBreevoortDalzielMaas2008,%
Maas2009PhysD,%
HazewinkelTsimitriMaasDalziel2010,HazewinkelGrisouardDalziel2011,%
ScolanErmDauxois2013}).

In the present paper, exact solutions to the problem (\ref{0gl.1p.16}--\ref{05gl.1p.20})
were found in an explicit  form for a class of domains including trapezoid $D$.
The construction of the solutions corresponds to Sobolev's ideas
of investigating of this problem and is based on the construction of
piecewise constant functions which satisfy (\ref{10}) and (\ref{06gl.1p.18})
in a generalized sense\footnote{
For the first time, similar piecewise solutions
of (\ref{10}-\ref{06gl.1p.18}) were used by R. A. Alexandryan when
studying the problem (\ref{0gl.1p.16}--\ref{05gl.1p.20})
for some special class of domains
\cite{Alexandryan1949eng}.}.
The explicit form of the solutions of (\ref{0gl.1p.16}--\ref{05gl.1p.20})
makes it possible to construct some special previously unknown type of inertial waves and to investigate
their properties. These waves are not inertial modes:
they correspond to the continuous spectrum of inertial oscillations.
Main properties of these waves are described.
These waves are progressive; the main features of their energy transfer are described.
The existence of these solutions casts doubt on the existence of non-point
wave attractors in the tank described above.

The results of the paper were partially announced in \cite{TroitKonfTurb2013}.

\section{Exact Solutions}
Denote by $H^1_0(D)$ the subspace of the Sobolev space $H^1(D)$
(see, for example, \cite{KopachevskyKrein2001engV1}, p. 34)
consisting of functions vanishing on the boundary $\partial D$, with the inner product
\begin{equation}\label{05gl.1p.22}
(f,g)_{1}=\int\!\!\!\int\limits_D\left(
\nabla f,\nabla \overline {g}\right) dx\,dz.
 \end{equation}
 On the interval $(0, + \infty)$,
consider functions $\psi(x,z;\,t)$ with values in $H^1_0(D)$ such that
$\psi_{t}$, $\psi_{tt}$ belong to $H^1_0(D)$
(by the derivative $\psi_{t}$, we mean the limit of $\Delta \psi /\Delta t $ as $\Delta t \to 0 $
in $H^1_0(D)$--norm). Suppose $\psi_0,\,\psi_1$ belong to $H^1_0(D)$.
As usual, $\psi(x,z;\,t)$ is called
a \emph{generalized solution} to (\ref{0gl.1p.16}--\ref{05gl.1p.20})
if it satisfies (\ref{05gl.1p.20}), and
for any smooth compactly supported in $D$ function $\varphi$, the equality
\begin{equation}\label{05gl.1p.24}
\int\limits_D(\psi_{ttx}\varphi_x+\psi_{ttz}\varphi_z+\psi_z\varphi_z) dxdz=0
\end{equation}
holds  for all $t>0$.
It is proved in \cite{KopachevskyKrein2001engV1,Ralston73,Sob54eng}
that just these solutions are physically meaningful.
If a generalized solution $\psi$ is sufficiently smooth, then it is a classical solution.

Consider the operator $\mathbf{A}_0:H^1_0(D)\rightarrow H^1_0(D)$ defined on smooth
compactly supported in $D$
functions as the solution to the problem
\begin{equation}\label{defA0}
\Delta\mathbf{A}_0 h = h_{zz},\quad
\bigl.\mathbf{A}_0 h \bigr|_{\partial D}=0,
\end{equation}
and its graph closure $\mathbf{A}:H^1_0(D)\rightarrow H^1_0(D)$.
It is well known that regardless of the shape of $D$ the operator
$\mathbf{A}$ is a bounded self-adjoint operator and the spectrum of
$\mathbf{A}$ is $[0,\,1]$ (see, e.g., \cite{Ralston73}).
Using $\mathbf{A}$ the problem (\ref{0gl.1p.16}-\ref{05gl.1p.20})
may be written in the form:
\begin{equation}%\label{0gl.1p.16}
 \psi_{tt}=-\mathbf{A} \psi, \quad \psi|_{t=0}=\psi_0,\quad \psi_t|_{t=0}=\psi_1.
\end{equation}

It is easy to calculate that the equations of the sides $z=z_i(x,\lambda)$ of $S(\lambda)$ are:
\[
CA:\quad z=1+\frac xa+\frac 1a-a,\quad
AB:\quad z=1-\frac xa-\frac 1a-a,
\]
\[
BD:\quad z=\frac xa+\frac 1a+a-3,\quad
DC:\quad z=-\frac xa-\frac 1a+a-1,
\]
where
$a=\sqrt{\frac{\lambda}{1-\lambda}}$, $\lambda \in (\frac 12, \frac 45)$
(see Figure \ref{fig:ma1}a)).

Denote by $C^1[c,d]$ the class of functions that are continuous together
with their first derivatives on an interval $[c,d]$.
\begin{theorem}\label{teo:1.1}
For each $\lambda \in (\frac 12,\frac 45)$,
 define the function
\begin{equation} %\label{12}
\chi(x,z,\lambda):=\left\{
\begin{array}{ll}
1, & \quad (x,z)\in P(\lambda),\\
0, & \quad (x,z)\in D\setminus P(\lambda).
\end{array}
\right.
 \end{equation}
Then for any $\lambda_*,\,\lambda_ {**}$ satisfying the condition
\begin{equation}\label{12v45}
\frac 12<\lambda_*<\lambda_{**}<\frac 45
 \end{equation}
and for any $\sigma(\lambda)\in C^1[\frac 12,\frac 45]$
\footnote{Of course, this condition can be weakened. },
 the function
\begin{equation}\label{V}
\Upsilon(x,z;\sigma;\lambda_*,\lambda_{**}):=\int\limits_{\lambda_*}^{\lambda_{**}}
\sigma(\mu)
\chi(x,z,\mu) d\mu
\end{equation}
belongs to $H_0^1(D)$.
\end{theorem}
\begin{proof}
Define the functions $\lambda_i(x,y)$, $i=1,\,2,\,3,\,4$, by the formulas:
 \begin{equation}\label{1s.10}
\lambda_1(x,z):=\frac {(z-1)^2+2x^2+6x+4+(1-z)\sqrt{(z-1)^2+4(x+1)}}{2((z-1)^2+(x+2)^2 )},
 \end{equation}
 \begin{equation}\label{1s.12}
\lambda_2(x,z):=\frac {(z-1)^2+2x^2+2x+(1-z)\sqrt{(z-1)^2-4(x+1)}}{2((z-1)^2+x^2 )},
 \end{equation}
\begin{equation}\label{1s.14}
\lambda_3(x,z):=\frac {(z+3)^2+2x^2+2x+(z+3)\sqrt{(z+3)^2-4(x+1)}}{2((z+3)^2+x^2 )},
 \end{equation}
\begin{equation}\label{1s.16}
\lambda_4(x,z):=\frac {(z+1)^2+2x^2+6x+4+(1+z)\sqrt{(z+1)^2+4(x+1)}}{2((z+1)^2+(x+2)^2 )}.
 \end{equation}
One can verify that $z_i(x,\lambda_i(x,z))\equiv z$, $i=1,\,2,\,3,\,4$.

Under the condition (\ref{12v45}) the parallelograms $S(\lambda_*)$, $S(\lambda_{**})$
decompose  $D$ into subdomains $D_i$, $i=1,2,...\,13$ (see Figure \ref{fig:ma1}b)).
 If
\[
G(\lambda):=\int_{\frac 12}^{\lambda} \sigma(\mu)\,d\mu
\]
is the primitive function, then
$\Upsilon(x,z;\sigma;\lambda_*,\lambda_{**})$ coincides with the function shown at
\textbf{Figure \ref{fig:ma2}}.
\begin{figure}[ht]
\begin{center}
\includegraphics[height=6.5cm]{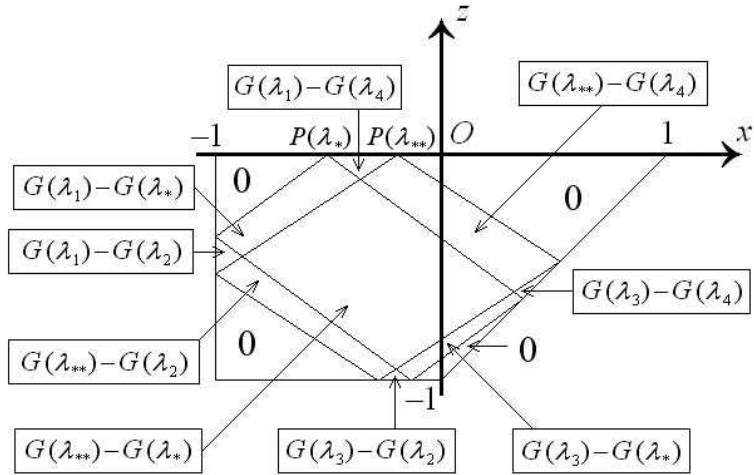}
\caption{The piecewise presentation of the function $\Upsilon(x,z;\sigma;\lambda_*,\lambda_{**})$.}
\label{fig:ma2}
\end{center}
\end{figure}
It is easy to proof that under the condition (\ref{12v45})
 $\Upsilon(x,z;\sigma;\lambda_*,\lambda_{**})$ is continuous and piecewise smooth
in the closure $\overline{D}$. And it is obvious that $\Upsilon(x,z;\sigma;\lambda_*,\lambda_{**})$
vanishes on the boundary $\partial D$.
\end{proof}

Unless otherwise specified, below we assume  (\ref{12v45}) to be fulfilled.

\begin{theorem}\label{teo:1.2}
Suppose that in (\ref{0gl.1p.16}--\ref{05gl.1p.20})
\[
{\psi}_0(x,z):=
\int\limits_{\lambda_{*}}^{\lambda_{**}}
\sigma_0(\lambda)\,\chi(x,z;\lambda)\,d\lambda,\quad
{\psi}_1(x,z):=
\int\limits_{\lambda_{*}}^{\lambda_{**}}
\sigma_1(\lambda)\,\chi(x,z;\lambda)\,d\lambda,
\]
where the functions $\sigma_i(\lambda)\in C^1\left[\frac 12,\frac 45\right]$, $i=1,2$.
Then the function
\[
{\psi}(x,z;t):=\!\!
\int\limits_{\lambda_{*}}^{\lambda_{**}}  \! \cos(\sqrt{\lambda}\; t)
\,\sigma_0(\lambda)\,\chi(x,z;\lambda)\,d\lambda
\]
\begin{equation}\label{typesol1}
+\int\limits_{\lambda_{*}}^{\lambda_{**}}   \frac{\sin(\sqrt{\lambda}\;t)}{\sqrt{\lambda}}\,
\sigma_1(\lambda)\,\chi(x,z;\lambda)\,d\lambda,
\end{equation}
is a solution to (\ref{0gl.1p.16}--\ref{05gl.1p.20}).
\end{theorem}
\begin{proof}
For simplicity assume $\sigma_1\equiv 0$.
For any smooth function $\varphi$ which is compactly supported in $D$,
we have
\[
\int\limits_D\!\!(\psi_{ttx}\varphi_x+\psi_{ttz}\varphi_z+\psi_z\varphi_z) dxdz=\!\!\!
\int\limits_D\!\!\left[\varphi_x\left(\int\limits_{\lambda_*}^{\lambda_{**}}
\!\!\cos(\sqrt{\lambda}\; t)
\,\sigma_0(\lambda)\,\chi(x,z;\lambda)\,d\lambda\right)_{\!\!\!ttx}\right.
\]
\[
+ \varphi_z\left(\int\limits_{\lambda_*}^{\lambda_{**}}
\cos(\sqrt{\lambda}\; t)
\,\sigma_0(\lambda)\,\chi(x,z;\lambda)\,d\lambda\right)_{\!\!\!ttz}
\]
\[
+ \left.\varphi_z\left(\int\limits_{\lambda_*}^{\lambda_{**}}
\cos(\sqrt{\lambda}\; t)
\,\sigma_0(\lambda)\,\chi(x,z;\lambda)\,d\lambda\right)_{\!\!\!z}\right]dxdz=
\]
\[
=-\int\limits_D\left[\varphi_x\left(\int\limits_{\lambda_*}^{\lambda_{**}}
\cos(\sqrt{\lambda}\; t)\lambda
\,\sigma_0(\lambda)\,\chi(x,z;\lambda)\,d\lambda\right)_{\!\!\!x}\right.
\]
\[
+ \varphi_z\left(\int\limits_{\lambda_*}^{\lambda_{**}}
\cos(\sqrt{\lambda}\; t)\lambda
\,\sigma_0(\lambda)\,\chi(x,z;\lambda)\,d\lambda\right)_{\!\!\!z}
\]
\[
- \left.\varphi_z\left(\int\limits_{\lambda_*}^{\lambda_{**}}
\cos(\sqrt{\lambda}\; t)
\,\sigma_0(\lambda)\,\chi(x,z;\lambda)\,d\lambda\right)_{\!\!\!z}\right]dxdz=
\]
\[
=\int\limits_D\left[\varphi_{xx}\int\limits_{\lambda_*}^{\lambda_{**}}
\cos(\sqrt{\lambda}\; t)\lambda
\,\sigma_0(\lambda)\,\chi(x,z;\lambda)\,d\lambda\right.
\]
\[
+\varphi_{zz}\int\limits_{\lambda_*}^{\lambda_{**}}
\cos(\sqrt{\lambda}\; t)\lambda
\,\sigma_0(\lambda)\,\chi(x,z;\lambda)\,d\lambda
\]
\[
-
\left.\varphi_{zz}\int\limits_{\lambda_*}^{\lambda_{**}}
\cos(\sqrt{\lambda}\; t)
\,\sigma_0(\lambda)\,\chi(x,z;\lambda)\,d\lambda \right]dxdz=
\]
\[
=\int\limits_D\left(\int\limits_{\lambda_*}^{\lambda_{**}}\left[\lambda\varphi_{xx}+
\lambda\varphi_{zz}-\varphi_{zz}\right]
\cos(\sqrt{\lambda}\; t)
\,\sigma_0(\lambda)\,\chi(x,z;\lambda)\,d\lambda\right) dxdz=
\]
\[
=\int\limits_{\lambda_*}^{\lambda_{**}}\left(\int\limits_D\left[\lambda\varphi_{xx}-
(1-\lambda)\varphi_{zz}\right]
\cos(\sqrt{\lambda}\; t)
\,\sigma_0(\lambda)\,\chi(x,z;\lambda)\,dxdz\right) d\lambda=
\]
\[
=\int\limits_{\lambda_*}^{\lambda_{**}}\left(\int\limits_{P(\lambda)}\left[\lambda\varphi_{xx}-
(1-\lambda)\varphi_{zz}\right]
\cos(\sqrt{\lambda}\; t)
\,\sigma_0(\lambda)\,dxdz\right) d\lambda.
\]
Consider the inner integral. Fix $\lambda$ and apply Green's formula:
\[
\int\limits_{P(\lambda)}\left[\lambda\varphi_{xx}-
(1-\lambda)\varphi_{zz}\right]
\cos(\sqrt{\lambda}\; t)
\,\sigma_0(\lambda)\,dxdz=
\]
\[
=-\int\limits_{S(\lambda)}\cos(\sqrt{\lambda}\; t)
\,\sigma_0(\lambda)
\left[\lambda\varphi_{x}\cos \gamma_1 -
(1-\lambda)\varphi_{z}\cos \gamma_2 \right]
\,dS(\lambda),
\]
where $S(\lambda)$ is oriented counterclockwise. We have (see Fig.~\ref{fig:ma1} a):
\[
\int\limits_{S(\lambda)}\cos(\sqrt{\lambda}\; t)
\,\sigma_0(\lambda)
\left[\lambda\varphi_{x}dz +
(1-\lambda)\varphi_{z}dx \right]=
\]
\[
=\int\limits_{CA}\cos(\sqrt{\lambda}\; t)
\,\sigma_0(\lambda)\left[
\lambda\varphi_{x}dz +
(1-\lambda)\varphi_{z}dx \right]+\int\limits_{AB}...+\int\limits_{BD}...+\int\limits_{DC}...
\]
Consider the integral over the segment $CA$. Denote by $\widetilde{\varphi}:=\varphi|_{CA}$.
As $dz=a^{-1}dx$ on the segment $CA$, it is easy to conclude that
\[
\int\limits_{CA}\cos(\sqrt{\lambda}\; t)
\,\sigma_0(\lambda)\left[\lambda\varphi_{x}dz +
(1-\lambda)\varphi_{z}dx\right]=
\]
\[
=\int\limits_{CA}\lambda\cos(\sqrt{\lambda}\; t)
\,\sigma_0(\lambda)\left[\varphi_{x}dz +\frac 1{a^2}\,\varphi_{z}dx\right]=
\]
\[
=\int\limits_{C}^{A}\lambda\cos(\sqrt{\lambda}\; t)
\,\sigma_0(\lambda)\left[\varphi_{x}\frac 1{a}\,dx +\frac 1{a^2}\,\varphi_{z}\,dx\right]=
\frac{\lambda}{a}\cos(\sqrt{\lambda}\; t)
\,\sigma_0(\lambda)\int\limits_{C}^{A}\frac {d{\widetilde{\varphi}}}{dx}dx=
\]
\[
=\frac{\lambda}{a}\cos(\sqrt{\lambda}\; t)
\,\sigma_0(\lambda)\left[\widetilde{\varphi}|_{A}-\widetilde{\varphi}|_{C}\right]=0.
\]
The other parts of the boundary $S(\lambda)$ can be considered similarly.
So we obtain
\[
\int\limits_D(\psi_{ttx}\varphi_x+\psi_{ttz}\varphi_z+\psi_z\varphi_z) dxdz=0.
\]
\end{proof}
\begin{corollary}\label{teo:1.3}
 Consider the primitive function
\begin{equation}\label{prim1}
F(\lambda,t):=\int\limits_{\lambda_*}^{\lambda}\left( \sigma_0(\mu)\cos (\sqrt{\mu}\,t)
+\sigma_1(\mu)    \frac{\sin(\sqrt{\mu}\;t)}{\sqrt{\mu}}\right) \,d\mu.
\end{equation}
Then the function $\psi(x,z,t;\lambda_*,\lambda_{**})$ defined at
\textbf{Figure \ref{fig:ma2s}}
is a solution to (\ref{0gl.1p.16}--\ref{0gl.1p.18})
in $D$.
\end{corollary}
\begin{figure}[ht]
\begin{center}
\includegraphics[height=6.5cm]{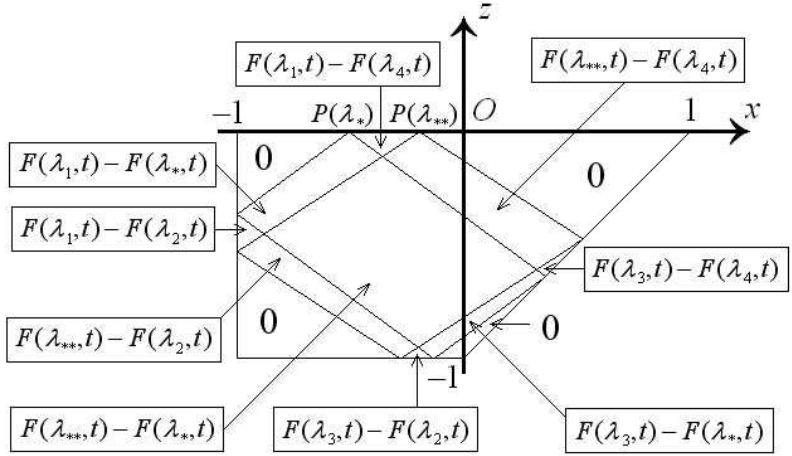}
\caption{The piecewise presentation of $\psi(x,z,t;\lambda_*,\lambda_{**})$.}
\label{fig:ma2s}
\end{center}
\end{figure}
\begin{example}\label{teo:1.3.1}
\rm
The function $\psi(x,z,t;\lambda_*,\lambda_{**})$ presented at
\textbf{Figure \ref{fig:ma3s}}
is a solution to the problem (\ref{0gl.1p.16}--\ref{05gl.1p.20})
in $D$.
\begin{figure}[ht]
\begin{center}
\includegraphics[height=11.5cm]{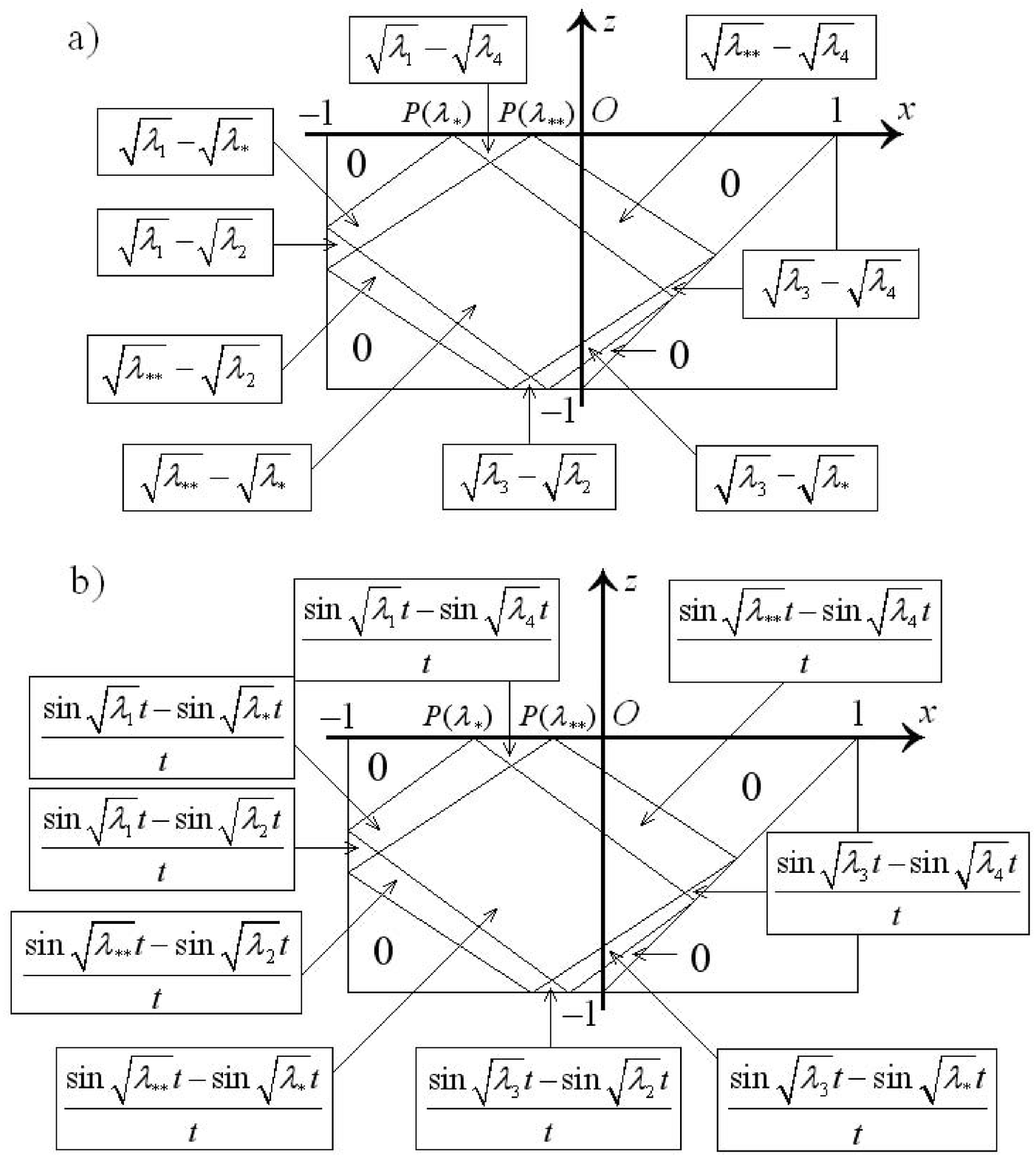}
\caption{a) The function $\psi_0(x,z)$ belongs to $H_0^1(D)$.
b) The solution $\psi(x,z,t;\lambda_*,\lambda_{**})$ of (\ref{0gl.1p.16}--\ref{05gl.1p.20})
corresponding to $\psi_0(x,z)$  and $\psi_1\equiv 0$.
}
\label{fig:ma3s}
\end{center}
\end{figure}
This follows immediately from Theorem \ref{teo:1.2} and Corollary
\ref{teo:1.3} with
 $\sigma_0=\frac 1{2\sqrt{\lambda}}$ and   $\sigma_1\equiv 0$.
\end{example}

It is clear that varying   $\lambda_*$, $\lambda_{**}$, and $\sigma_i$,
we get inertial waves with very
different properties. Choosing appropriate functions
$\sigma_i$ it is possible to construct arbitrarily smooth solutions of this type.

\begin{example}\label{teo:1.3.2}
\rm If
\[
F(\lambda, t)=\frac {\sin(\sqrt\lambda (c-t)-d)}{2(c-t)}+
\frac {\sin(\sqrt\lambda (c+t)-d)}{2(c+t)}+
\frac {\sin\sqrt\lambda\, t}{t} \quad \mbox{and}\,\,
\]
\[
F_1(\lambda)=\frac {\sin(c\sqrt\lambda -d)}{c}+\sqrt\lambda,
\]
where
\[
c=\frac {2\pi}{\sqrt\lambda_{**}-\sqrt\lambda_*},\quad
d=\frac {\pi(\sqrt\lambda_{**}+\sqrt\lambda_*) }{\sqrt\lambda_{**}-\sqrt\lambda_*},
\]
then the function $\psi(x,z,t;\lambda_*,\lambda_{**})$
defined  at Figure \ref{fig:ma2s}, is a classical solution to the problem
(\ref{0gl.1p.16} --- \ref{05gl.1p.20}) in  $D$,
where $\psi_0$ is given at \textbf{Figure  \ref{fig:ma5}}, and $\psi_1\equiv 0$.
\begin{figure}[ht]
\begin{center}
\includegraphics[height=6cm]{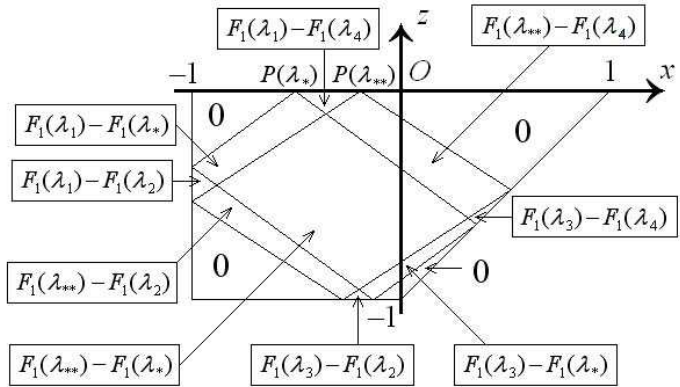}
\caption{The initial state of the classical solution $\psi(x,z,t;\lambda_*,\lambda_{**})$. }
\label{fig:ma5}
\end{center}
\end{figure}
One can prove this by a direct calculation, but it follows from the fact that
in this case $\sigma_1\equiv 0$ and
the support of $\sigma_0$ lies within the interval $[\lambda_*,\lambda_{**}]$:
\begin{equation} \label{13ab}
\sigma_0=\left\{
\begin{array}{ll}
0,       &\mbox{ for } \lambda\leq \lambda_*,\\
\frac 1{2\sqrt{\lambda}}\left(\cos\left(\frac {2\pi\sqrt{\lambda}}
{\sqrt{\lambda_{**}}-\sqrt{\lambda_*}}-
\frac {\pi(\sqrt{\lambda_{**}}+\sqrt{\lambda_*} )}
{\sqrt{\lambda_{**}}-\sqrt{\lambda_*}}\right)+1\right),
&\mbox{ for }  \lambda_*<\lambda\leq\lambda_{**}, \\
0,&\mbox{ for }  \lambda_{**}<\lambda . \\
\end{array}
\right.
 \end{equation}
\end{example}

 Now using the obtained stream function $\psi$
  we can reconstruct the corresponding
 class of inertial waves $ \mathbf{U}$
 (the solutions to (\ref{0gl.1p.4}--\ref{init1})).
This procedure is commonly known.
The components $u$ and $w$ can be found directly from the equations
that determine the stream function:
\[
u=-\frac {\partial {\psi}}{\partial z},\quad
w=\frac {\partial {\psi}}{\partial x}.
\]
According to the second equation in (\ref{0gl.1p.4})
for the component $v$ we have
\begin{equation} \label{2002f1}
v=\int_0^t \frac {\partial \psi}{\partial z}(x,z,\tau;\lambda_*,\lambda_{**})\,d\tau
+v_0(x,z),
 \end{equation}
where $v_0(x,z)$ is a solution of the equation
\begin{equation} \label{2002f2}
\frac {\partial v_0}{\partial z}=-\Delta \psi_1(x,z).
 \end{equation}
It is clear that (\ref{2002f2})  determines $v_0(x,z)$
 only up to an arbitrary function $v_*(x)$
(of course, we are interested in the case $v \in L_2(D)$).
The function $p(x,z,t)$ can be found from the system of equations
\begin{equation} \label{2002f3}
\frac {\partial p}{\partial x}=
\int_{0}^{t}\frac{\partial {\psi}}{\partial z}(x,z;s)\,ds
+\frac {\partial }{\partial t}\frac{\partial {\psi}}{\partial z}(x,z;t)+v_0(x,z),
 \end{equation}
\begin{equation} \label{2002f4}
\frac {\partial p}{\partial z}=
-\frac {\partial }{\partial t}\frac{\partial {\psi}}{\partial x}(x,z;t).
 \end{equation}
 Respectively  $p(x,z,t)$ is determined up to the function $p_*(x)=\int v_*(x)dx$.
The functions
\[
\mathbf{U}_*=(0,v_*(x),0),\quad p_*(x)=\int v_*(x)dx
\]
are the stationary solutions of (\ref{0gl.1p.4}--\ref{init1}) (see \cite{Fokin2002}).

It is obviously that the conventional characteristics
(frequency, amplitude, wave vector dispersion relation,
the direction of propagation of energy, and so on)
can not be applied to the new found  class of inertial waves.

In the next sections in order to emphasize
 main features of the new type of inertial waves
we consider the case $\psi_1\equiv 0$ and $v_0\equiv 0 $.
Thus these inertial waves are
$\mathbf{U}= \mathbf{U}(x,z,t;\lambda_*,\lambda_{**})=(u,\,v,\,w)$, where
\begin{equation}\label{1s.50}
u=-\frac {\partial {\psi}}{\partial z},\quad
v=\int\limits_{0}^{t}\frac{\partial {\psi}}{\partial z}(x,z;s)\,ds,\quad
w=\frac {\partial {\psi}}{\partial x},\quad u,\,v,\,w \in L_2(D),
\end{equation}
and $\psi=\psi(x,z,t;\lambda_*,\lambda_{**})$ corresponds to the case
$\sigma_1\equiv 0$.

\section{Geometric properties of solutions
$\psi(x,z,t;\lambda_*,\lambda_{**})$ and $\mathbf{U}(x,z,t;\lambda_*,\lambda_{**})$}

\begin{theorem}\label{teo:2.1}
$\mathbf{U}(x,z,t;\lambda_*,\lambda_{**})\equiv
\mathbf{0}$ in the domains $D_i$, $i=9 - 13$ (see Figure \ref{fig:ma1}b)).
\end{theorem}
\begin{proof}
In the domains $D_i$, $i=10 - 13$ the function  $\psi\equiv 0$ according
to the construction. In $D_9$ the function $\psi$ does not depend on spatial variables,
hence $\mathbf{U}(x,z,t;\lambda_*,\lambda_{**})\equiv
\mathbf{0}$ in it too.
\end{proof}
\begin{remark}\label{rk:2.2}
It is easy to conclude that if $ \lambda_* $
and $\lambda_{**} $  tend to their extreme values,
i.e., $\lambda_*\rightarrow \frac 12 $ and $\lambda_{**}\rightarrow \frac 45 $,
then the oscillations affect almost the entire region  $ D $
(see Fig.~\ref{fig:ma6} a)).
But if $ \lambda_* $ and $ \lambda_{**} $ are
close to each other  and tend to some value $\frac 12< \lambda_{***} <\frac 45$, then
the region, where the oscillations take place, ``tends
to form the parallelogram $S(\lambda_{***})$''.
The rest of the fluid is not exposed to these oscillations (see
\textbf{Figure \ref{fig:ma6}}b)).
\begin{figure}[ht]
\begin{center}
\includegraphics[height=4.0cm]{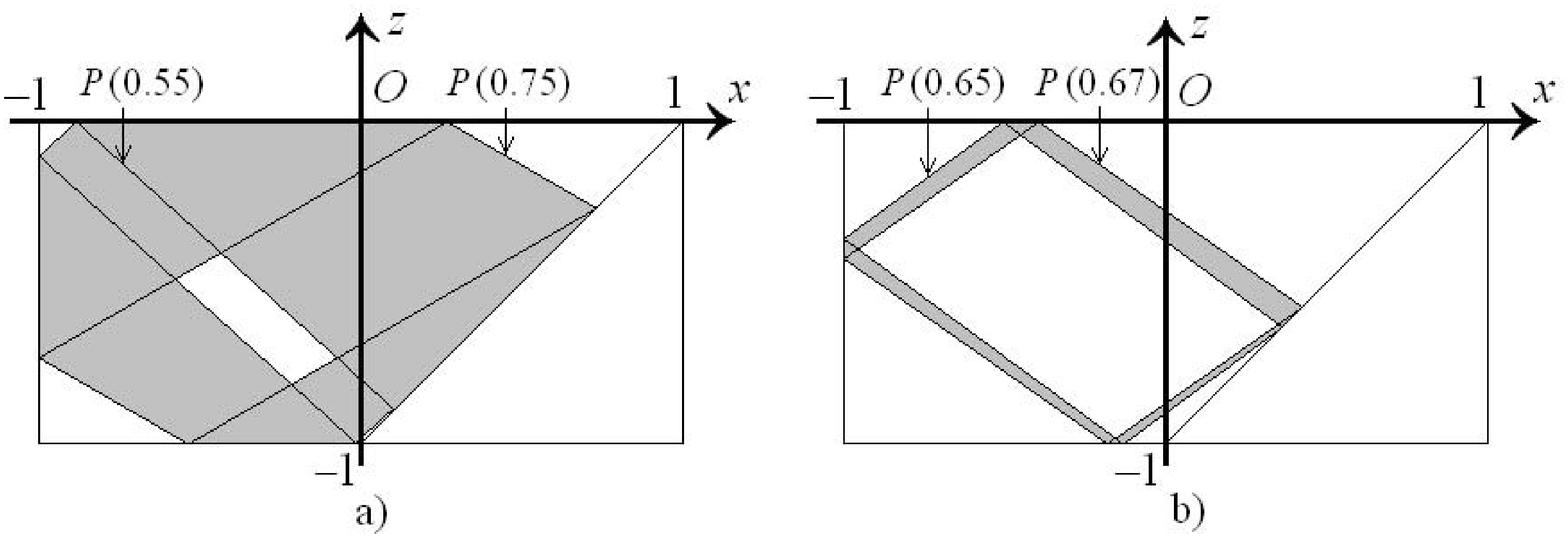}
\caption{
The comparison of areas affected by the fluid oscillations.
\newline
a) If $\lambda_*\rightarrow \frac 12 $ and $\lambda_{**}\rightarrow \frac 45 $
respectively, then the oscillations affect almost the entire domain  $ D $.
\newline
b) If $ \lambda_* $ and $ \lambda_{**} $ are
close to each other  and tend to each other, then
the region,  where the oscillations of the fluid take place, tends
to form a parallelogram (compare with Figure \ref{fig:ma2s}, where
the values $\lambda_* = 0.65$ and $ \lambda_{**} = 0.7$ are taken).
}
\label{fig:ma6}
\end{center}
\end{figure}
\end{remark}
\begin{remark}\label{rk:2.3}
The construction of the solutions $\psi(x,z,t;\lambda_*,\lambda_{**})$ and
$\mathbf{U}(x,z,t;\lambda_*,\lambda_{**})$ makes it obvious
that the oscillations of this type  exist not only in the trapezoid
$ D $, but in some rather wide class of domains.
Namely, if for some fixed values $ \lambda_{*} $ and $ \lambda_{**} $, we
arbitrarily vary the boundary  $\partial D $
in those its parts which are adjacent to areas where $ \psi(x, z, t) \equiv 0 $, then
the function $\mathbf{U}(x,z,t;\lambda_*,\lambda_{**})$ will be a solution to
(\ref{0gl.1p.4}--\ref{0gl.1p.10}) in the new domain too (see
\textbf{Figure \ref{fig:ma7}}).

This extension of the class of domains is important, because natural waters have
various configurations. The geometric properties described above mean that
when the oscillations of this type occur, the entire volume of the water
is divided into areas of very intense oscillations and areas of complete rest.
Probably, the latter phenomenon is the cause of the ``strange behavior'' of fish in the ocean,
when a shoal is swimming in a certain direction and then turns suddenly,
as if it was faced with the invisible flat wall: most likely,
a zone of intense oscillations begins behind this ``wall''.
\end{remark}
\begin{remark}\label{rk:2.4}
In Remark~\ref{rk:2.3} it is shown how
one can construct domains where the inertial waves of this type may occur, on the base of $D$.
Of course, the trapezoid  $D$ is not a distinguished domain: there are a lot of such domains.
A complete description of their configurations
is difficult and lays beyond the scope of this article.
\begin{figure}[ht]
\begin{center}
\includegraphics[height=3.5cm]{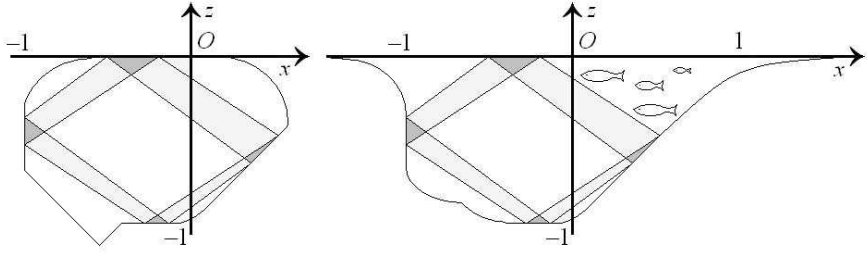}
\caption{Some other possible configurations $D$ where this type of inertial waves can occur.}
\label{fig:ma7}
\end{center}
\end{figure}
\end{remark}

\begin{theorem}\label{teo:2.5}
Suppose that  $\mu_*<\mu_{**}$,
$\mu_*,\mu_{**} \in (\frac 12,\frac 45)$,
 $(\mu_*,\mu_{**}) \cap (\lambda_*, \lambda_{**}) = \emptyset$
 and $\widetilde{\psi}(x,z,t;\mu_*,\mu_{**})$ is a solution
%to the problem (\ref{0gl.1p.16}--\ref{05gl.1p.20})
of the considered type (\ref{typesol1})  with some
functions $\widetilde{\sigma}_i(\lambda)\in C^1\left[\frac 12,\frac 45\right]$, $i=1,2$.
\footnote{In this theorem we do not need the restriction $\sigma_1 \equiv 0$.}
Then  the equality
\[
(\psi(x,z,t;\lambda_*,\lambda_{**}),
\widetilde{\psi}(x,z,t;\mu_*,\mu_{**}))_1=
\]
\begin{equation}
=\int\limits_D\left(
\nabla \psi(x,z,t;\lambda_*,\lambda_{**}),
\nabla \overline {\widetilde{\psi}(x,z,t;\mu_*,\mu_{**})}\right) dxdz=0
\end{equation}
is true for all $t>0$.
\footnote{This property is analogous to the orthogonality properties of the inertial modes
which correspond to distinct frequencies (see \cite{GreenspanBook68}, p. 52).}
\end{theorem}
\begin{proof}
On the interval $(\frac 12,\, \frac 45)$
for an arbitrary function $\sigma(\lambda)\in C^1[\frac 12,\,\frac 45]$,
define the function $\Psi(\lambda)$ with values in $H^1_0(D)$:
 \begin{equation} \label{13a}
\Psi(\lambda)= \Psi(x,z;\sigma;\lambda):=\left\{
\begin{array}{ll}
\bf{0},       &\mbox{ for } \lambda\leq \lambda_*,\\
\int\limits_{\lambda_*}^{\lambda}\sigma (\mu)\chi(x,z;\mu)d\mu,
&\mbox{ for }  \lambda_*<\lambda\leq\lambda_{**} ,\\
\int\limits_{\lambda_*}^{\lambda_{**}}\sigma (\mu)\chi(x,z;\mu)d\mu,
&\mbox{ for }  \lambda_{**}<\lambda , \\
\end{array}
\right.
 \end{equation}
where $\bf{0}$ stands for the zero element of $H^1_0(D)$.
Consider the function
\begin{equation}\label{11a}
\Xi:=\mathbf{A}\left(\Psi(\nu_2)-\Psi(\nu_1)\right)-\int\limits_
{\nu_1}^{\nu_2}\lambda\,d \Psi(\lambda),
\end{equation}
where the operator $\mathbf{A}$ is defined in Sec. 1, and
$ \nu_1$, $\nu_2 $ are arbitrary in $(\frac 12,\, \frac 45)$.
For $\lambda_*\leq \nu_1\leq\nu_2\leq \lambda_{**}$, we have:
\[
\Xi=\mathbf{A}\,\int\limits_{\nu_1}^{\nu_2}
\sigma(\lambda) \;\chi(x,z;\lambda) \,d\lambda -
\int\limits_{\nu_1}^{\nu_2}\lambda\;
\sigma(\lambda) \;\chi(x,z;\lambda) \,d\lambda .
\]
We prove now that $\Xi=\bf{0}$. This is equivalent to prove that
\[
\int\limits_{D}\left\{
\frac{\partial^2}{\partial z^2} \,\int\limits_{\nu_1}^{\nu_2}
\sigma(\lambda) \;\chi(x,z;\lambda) \,d\lambda  \right.
-\Delta\left.
\int\limits_{\nu_1}^{\nu_2}\lambda\,
\sigma(\lambda) \;\chi(x,z;\lambda) \,d\lambda  \right\}\cdot \varphi\,dxdz =0,
\]
where $\varphi$ is an arbitrary smooth function compactly supported in $D$ (the last statement means that
$\Delta\Xi$ is zero element in the space $H^{-1}(D)$,
see, e.g., \cite{KopachevskyKrein2001engV1}). We have:
\[
\int\limits_{D}\left\{
\frac{\partial^2}{\partial z^2} \,\int\limits_{\nu_1}^{\nu_2}
\sigma(\lambda) \;\chi(x,z;\lambda) \,d\lambda  \right.
-\Delta\left.
\int\limits_{\nu_1}^{\nu_2}\lambda\,
\sigma(\lambda) \;\chi(x,z;\lambda) \,d\lambda  \right\}\cdot \varphi\,dxdz =
\]
\[
=\int\limits_{D} \left\{
 \int\limits_{\nu_1}^{\nu_2}
\sigma(\lambda) \;\chi(x,z;\lambda) \,d\lambda \cdot
\frac{\partial^2 \varphi}{\partial z^2}
-\int\limits_{\nu_1}^{\nu_2}\lambda\;
\sigma(\lambda) \;\chi(x,z;\lambda) \,d\lambda \cdot
\Delta \varphi \right\}dxdz =
\]
\[
=\int\limits_{D}  \int\limits_{\nu_1}^{\nu_2}\left\{
\sigma(\lambda) \;\chi(x,z;\lambda)  \cdot
\frac{\partial^2 \varphi}{\partial z^2}
-\lambda\,\sigma(\lambda) \;\chi(x,z;\lambda) \cdot
\Delta \varphi \right\} d\lambda \,dxdz =
\]
\[
= -\int\limits_{\nu_1}^{\nu_2} \int\limits_{D} \left\{
\lambda \varphi_{xx}- (1 -\lambda) \varphi_{zz}\;
 \right\}\chi(x,z;\lambda) \sigma(\lambda) dxdz\,d\lambda =
\]
\[
= -\int\limits_{\nu_1}^{\nu_2} \int\limits_{P(\lambda)}
\left\{
\lambda \varphi_{xx}- (1 -\lambda) \varphi_{zz}\;
 \right\} \sigma(\lambda) dxdz\,d\lambda=0
\]
(see the proof of Theorem \ref{teo:1.2}). It is clear that
 $\Xi=\bf{0}$ for arbitrary   $\nu_1, \nu_2 \in  (\frac 12,\, \frac 45)$ too.
Therefore there exist
a function $h \in H^1_0(D)$ such that $\Psi (\lambda)= E_{\lambda}h$, where
$E_{\lambda}$ is the resolution of the identity for $\mathbf{A}$.
As for each fixed $t>0$ every solution of the type
(\ref{typesol1}) can be presented in the form
\[
\psi(x,z,t;\lambda_*,\lambda_{**})=
\Psi (x,z;\sigma; \,\lambda_{**})-\Psi (x,z; \sigma; \,\lambda_{*}),
\]
with some appropriate $\sigma$,
then the statement of the theorem follows from the condition immediately.
\end{proof}

\begin{proposition}\label{teo:2.5r}
Denote by $H_{c}$ the closure of the linear span of
all functions of the form (\ref{13a}) in $H^1_0(D)$.
 The spectrum of the operator $\mathbf{A}$ is absolutely continuous on $H_{c}$:
$\sigma (\mathbf{A}|_{H_c})=\sigma_{ac}(\mathbf{A}|_{H_c})$.
\end{proposition}
\begin{proof}
We have to prove the absolute continuity of the function
$\Vert\Psi (\lambda)\Vert_{1}^2$.
For $\lambda_*\leq \lambda\leq \lambda_{**}$, we have:
\[
\Vert\Psi (\lambda)\Vert_{1}^2=
\Vert\Upsilon(x,z;\sigma;\lambda_*,\lambda)\Vert_{1}^2,
\]
(see (\ref{V})).
Denote by $D_i(\lambda)$ the domains corresponding to the
decomposition of $D$ by $S(\lambda_*)$ and $S(\lambda)$, $i=1,...,13$.
The properties of the function  $\Upsilon$
described in the proof of Theorem \ref{teo:1.1}
allow us to assert that
\[
\Vert\Upsilon(x,z;\sigma;\lambda_*,\lambda)\Vert_{1}^2=
\int_D \left|\nabla(\Upsilon(x,z;\sigma;\lambda_*,\lambda ))\right|^2 dx dz=
\]
\[
=\int_{D_1(\lambda)}\left|\nabla (G(\lambda_{1})-G(\lambda_{4}))\right|^2 dx dz +
\int_{D_2(\lambda)}
\left|\nabla (G(\lambda_{1})-G(\lambda_{2}))\right|^2 dx dz
\]
\[
+\int_{D_3(\lambda)}
\left|\nabla (G(\lambda_{3})-G(\lambda_{2}))\right|^2 dx dz +
\int_{D_4(\lambda)}
\left|\nabla (G(\lambda_{3})-G(\lambda_{4}))\right|^2 dx dz
\]
\[
+\sum\limits_{i=5}^8\int_{D_i(\lambda)}
\left|\nabla (G(\lambda_{i-4}))\right|^2 dx dz.
\]
Here the integrands are smooth and they
do not depend on $\lambda$, and all boundaries $\partial D_i(\lambda)$
 smoothly depend on $\lambda$. Thus the function
$\Vert\Psi (\lambda)\Vert_{1}^2$ is not only absolute continuous but smooth.
\end{proof}

The explicit form of the obtained
 inertial waves makes it possible to get animated graphics of the stream
function, of the velocity field of the fluid, of the energy density function, etc. (See Appendix)

\section{Energy properties of the solutions
$\psi(x,z,t;\lambda_*,\lambda_{**})$ and $\mathbf{U}(x,z,t;\lambda_*,\lambda_{**})$.}
All subdomains of $D$ that will be considered below are supposed to be
 Lipschitz domains (see definition in  \cite{KopachevskyKrein2001engV1}, p. 34).
\begin{definition}
\rm
Suppose that a domain $\Omega_0 \subset D$.
If a solution of the system (\ref{0gl.1p.4}--\ref{init1})
is such that for any subdomain $\Omega_1 \subset \Omega_0$ the
amount of kinetic energy  in $\Omega_1 $
\begin{equation}\label{0gl.1p.07b}
{\cal E}(t;\,\Omega_1):=\int\limits_{\Omega_1} \left(\left|u\right|^2+
\left|v\right|^2+\left|w\right|^2\right) dxdz={\rm const}
\end{equation}
(does not depend on time), then we say that the solution has the
\emph{property} $S$ in $\Omega_0$.
\footnote{
Perhaps it would be logical to call it ``standing wave in the domain $\Omega_0$'',
but this term, alas, is already used.}
\end{definition}
It is clear that each inertial mode has the property $S$ in the whole domain $D$.
\begin{definition}
\rm
If for a solution of (\ref{0gl.1p.4}--\ref{init1})
in the domain $D$ there exists a subdomain $\Omega$ such that
the amount of kinetic energy in $\Omega $
is not constant, then this solution is called a \textit{progressive} inertial wave.
\end{definition}
\begin{theorem}\label{teo:3.3}
The solution $\mathbf{U}(x,z,t;\lambda_*,\lambda_{**})$ has the property
$S$ in the domain $\bigcup\limits_{i=5}^{13} D_i $.
\end{theorem}
\begin{proof}
It is sufficient to prove that this solution has the property
$S$ in the domains $D_i$, $i=5,...,8$. Consider, for example,
the amount of kinetic energy corresponding to
some domain $\Omega_5\subset D_5$:
\begin{equation}\label{D5}
{\cal E}(t;\,\Omega_5)=\int_{\Omega_5} \left(u^2+
v^2+w^2\right) dxdz.
\end{equation}
Denote
\begin{equation}\label{dertdens}
\Lambda:=uu_t+vv_t+ww_t.
\end{equation}
We now show  that $\Lambda\equiv 0$ in $D_5$.
Consider the primitive function
\[
F(\lambda,t):=\int_{\lambda_*}^{\lambda}\sigma_0(\mu)\cos (\sqrt{\mu}\,t)\,d\mu.
\]
Then in the domain $D_5$
\[
u=-\bigl(F(\lambda_1,t)-F(\lambda_*,t)\bigr)_z=-\frac{\partial F( \lambda_1, t)}{\partial \lambda_1}\cdot
\frac{\partial \lambda_1}{\partial z},\quad
u_t=-\frac{\partial^2 F( \lambda_1, t)}{\partial t\, \partial \lambda_1}\cdot
\frac{\partial \lambda_1}{\partial z},
\]
\[
v=\int_{0}^{t}\frac{\partial F( \lambda_1, \tau)}{\partial \lambda_1}d\,\tau \cdot
\frac{\partial \lambda_1}{\partial z},\qquad
v_t=\frac{\partial F ( \lambda_1, t)}{ \partial \lambda_1}\cdot
\frac{\partial \lambda_1}{\partial z},
\]
\[
w=\frac{\partial F( \lambda_1, t)}{\partial \lambda_1}\cdot
\frac{\partial \lambda_1}{\partial x},\qquad
w_t=\frac{\partial^2 F( \lambda_1, t)}{\partial t\, \partial \lambda_1}\cdot
\frac{\partial \lambda_1}{\partial x},
\]
(see (\ref{1s.50}), Theorem \ref{teo:1.2}, and Corollary \ref{teo:1.3}).
We have
\[
\frac{\partial^2 F( \lambda_1, t)}{\partial t\, \partial \lambda_1}=
\bigl(\sigma_0(\lambda_1)\cos (\sqrt{\lambda_1}\,t)\bigr)_t=
-\sigma_0(\lambda_1)\sqrt{\lambda_1}  \sin (\sqrt{\lambda_1}t),
\]
\[
\int_{0}^{t}\frac{\partial F( \lambda_1, \tau)}{\partial \lambda_1}d\,\tau=
\int_{0}^{t}\sigma_0(\lambda_1)\cos (\sqrt{\lambda_1}\, \tau)\,d\,\tau=
\frac{\sigma_0(\lambda_1)}{\sqrt{\lambda_1} } \sin (\sqrt{\lambda_1}t).
\]
Therefore
\begin{equation}\label{Lambda}
\Lambda =\frac{\partial F( \lambda_1, t)}{\partial \lambda_1}\cdot{\sqrt{\lambda_1} }
\sigma_0(\lambda_1)\sin (\sqrt{\lambda_1}t)
\biggl[ \left(\frac{\partial \lambda_1}{\partial z}\right)^2\cdot \frac {1-\lambda_1}{\lambda_1}-
\left(\frac{\partial \lambda_1}{\partial x}\right)^2\biggr].
\end{equation}
According to the definition of the function $\lambda_1(x,z)$
\[
z_1(x,\lambda_1(x,z))\equiv z, \quad \mbox{where}\quad  z_1(x,\lambda)=1+\frac xa+\frac 1a -a, \quad a^2
=\frac {\lambda}{1-\lambda}.
\]
Denote
\[
M(x,z,a):=az-x-a-1+a^2.
\]
Using the implicit function theorem we have
\[
a_x=-\frac{\frac{\partial M}{\partial x}}{\frac{\partial M}{\partial a}}=-\frac{-1}{z-1+2a},
\quad
a_z=-\frac{\frac{\partial M}{\partial z}}{\frac{\partial M}{\partial a}}=-\frac{a}{z-1+2a}.
\]
As
\[
\frac{\partial\lambda}{\partial a}=\frac {2a}{(1+a^2)^2},
\]
we get
\[
\frac{\partial \lambda_1}{\partial x}=
\frac{\partial \lambda_1}{\partial a} \frac{\partial a}{\partial x}=
\left.\frac {2a}{(1+a^2)^2}\cdot \frac{1}{z-1+2a}\right|_{\lambda=\lambda_1},
\]
\[
\frac{\partial \lambda_1}{\partial z}=
\frac{\partial \lambda_1}{\partial a}\frac{\partial a}{\partial z}=
\left.\frac {2a}{(1+a^2)^2} \cdot \frac{-a}{z-1+2a}\right|_{\lambda=\lambda_1}.
\]
Therefore, the expression in square brackets in (\ref{Lambda})
is equal to zero; respectively $\Lambda\equiv 0$ in $D_5$; and
the amount ${\cal E}(t;\,\Omega_5)$ of kinetic energy in $\Omega_5$ does not depend on time.
\end{proof}
%\newpage
\begin{theorem}
\begin{equation}\label{1s.53}
{\cal E}(t;\, \bigcup\limits_{i=1}^{4} D_i)=\sum\limits_{i=1}^{4} \int\limits_{D_i} \left(\left|u\right|^2+
\left|v\right|^2+\left|w\right|^2\right)  dxdz={\rm const}.
\end{equation}
\end{theorem}
\begin{proof}
The statement  immediately follows from the energy conservation law (\ref{energyLaw2dim}) and
Theorem~\ref{teo:3.3}.
\end{proof}
However, it is easy to establish that within the domains $ D_i $, $i=1,2,3,4$,
the amount of the kinetic energy is not constant.
Thus, the inertial wave $\mathbf{U}(x,z,t;\lambda_*,\lambda_{**})$ is progressive.  It
``pumps'' energy from $D_i$ to $D_j$, when $i,j=1,2,3,4$, $i\neq j$,
but in other areas it behaves like a standing wave.
\begin{remark}\label{rk:3.5}
For all $t>0$ the energy density function is piecewise smooth in $D$ with respect
to spatial variables, and there is no point
in $\overline{D}$ in which it tends to infinity as  $t \rightarrow \infty$.
\end{remark}

\section{Discussion: which wave attractors do exist in the trapezoid?}
As it was mentioned above, papers of many authors studying inertial waves are devoted to the
emergence of wave attractors. But what are  \textit{wave attractors}?
A rigorous definition of a wave attractor does not exist until now,
but from many recent papers in Geophysics and Astrophysics it becomes clear that
the researchers mean that in the points of a wave attractor
the energy density of a fluid tends to infinity
(in some sense) as time goes to infinity.

 It follows from \cite{TroFirst99}  that
for the considered above trapezoid $D$ (see (\ref{trapezoidD}))
 there is no inertial mode
 corresponding to the frequency interval $\lambda \in (0, \frac 12)$.
\begin{figure}[ht]
\begin{center}
\includegraphics[height=4cm]{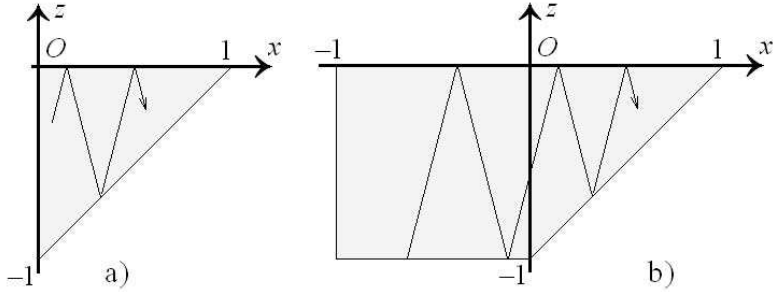}
\caption{a) In this triangle, the vertex (1,\,0) is a point wave attractor (see \cite{TroitRGMF10}).
b) In the trapezoid $D$, the vertex (1,\,0) is a possible point wave attractor.
}
\label{fig:ma6n}
\end{center}
\end{figure}
Most likely that for the trapezoid $D$, it is possible to construct a class of inertial waves, whose energy
in process of time will be concentrated in an arbitrarily small neighborhood of the vertex
{(1,\,0)} (see Figure \ref{fig:ma6n}b)), as it was done in the case of the triangle domain in
\cite{TroitRGMF10}. Therefore, with a high degree of certainty,
we may claim that the vertex {(1,\,0)} is a \textit{point wave attractor}
(of course, this fact requires a careful mathematical proof).

Now suppose $\lambda \in (\frac 12, \frac 45)$.
There are a lot of papers in Geophysics  where
the existence of wave attractors of the form
$S(\lambda)$ in the trapezoid $D$ is claimed
or is considered as a proven fact
(see, e.g.,
\cite{
MaasBenielliNature97,%
MaasFlMech2001,%
StaquetSommeria2002,%
MandersMaas2003,%
MandersMaas2004,%
MaasChaos2005,%
HazewinkelMaasDalziel2007,%
SwartSleijpenMaasBrandts2007,%
HarlanMaas07,%
Harland08p1,%
GerkemaZimmerman2008,%
GerkZimmMaasHaren2008,%
GrisouardStaquetPairaud2008,%
HazewinkelBreevoortDalzielMaas2008,%
LamMaas2008,%
Maas2009PhysD,%
HazewinkelTsimitriMaasDalziel2010,%
HazewinkelGrisouardDalziel2011,%
ScolanErmDauxois2013,%
GerkemaMaas2013,%
BrouzetErmDauxois2016}).

Alas, our opinion differs from the opinion of these authors.
Namely, we are convinced that as a result of their experiments the researchers
observed the oscillations of the fluid corresponding
to the solutions of the described above type for some
close to each other values $\lambda_*$ and $\lambda_{**}$ (see, for example,
\textbf{Figure \ref{fig:ma7n}}).
As it was shown above (Remarks \ref{rk:2.2}, \ref{rk:3.5}), these solutions
certainly do not create attractors.
\begin{figure}[ht]
\begin{center}
\includegraphics[height=9.5cm]{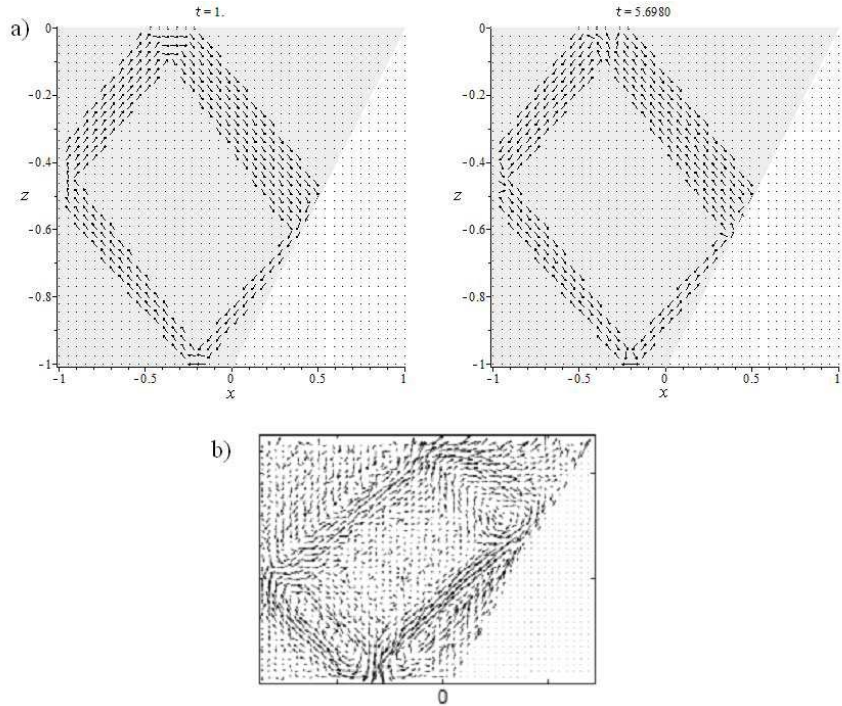}
\caption{
One can compare our field-plots of $[u,\, w]$ at different moments of time (see a))
with the velocity field in a stratified fluid in some moment (see b)) that was computed  on the base of
experimental measurements made by Particle Image Velocimetry
in \cite{HazewinkelMaasDalziel2007} (the original Figure 3a
in \cite{HazewinkelMaasDalziel2007}).}
\label{fig:ma7n}
\end{center}
\end{figure}

Of course, the question on the existence of inertial waves that can concentrate
their energy in an arbitrarily small neighborhood of $S(\lambda)$ currently remains open,
because to deny  them one must investigate some special properties of the constructed above
class of solutions, for example, their completeness in certain spaces of functions, etc.,
but these issues are complex and unlikely to be resolved soon.

\section*{Acknowledgments}
I would like to thank B. I. Sadovnikov for the helpful discussions at the beginning
of the present research and the participants of the International Conference
``Turbulence and Wave Processes'' (Moscow State University, 2013) for useful discussions of the obtained
results; and I also thank A.~E.~Troitskaya and V.~D.~Rusakov
for their help in the preparation of the graphs of the obtained solutions.

\begin{appendix}
\section{Some graphs}
The basic properties of the obtained solutions of (\ref{0gl.1p.4}--\ref{init1})
were described above, but taking into account that these solutions
are not almost-periodic functions, their particular properties
are also very interesting and
should be studied much more carefully. Here we present only
 a few simple graphs
of these functions at certain times\footnote{
All graphs were generated using the package Maple 16.
}.
We set
\[
\lambda_*=0.65,\quad \lambda_{**}=0.7
\]
and choose the function presented at the Figure \ref{fig:ma3s} as the solution $\psi(x,z,t)$;
the other graphs correspond to this solution too.
The graphs of the function $\psi(x,z,t)$ for small values of $t$
(up to 25) and for sufficiently large $t$ (up to 2500)
are presented separately.
This allows to notice a substantial difference in the behavior of the function
in these two cases.
Although we obtained the solutions of the linearized system of equations,
the study of their behavior should shed light on the formation of turbulent flows
in containers having the described above configuration:
as it can be seen from the graphs, the large-scale oscillations of the rotating fluid
are converted to the vortex motions, the scale of which decreases when $t \to \infty$;
respectively, the energy of  the initial  large-scale perturbation
is redistributed among the
micro-scale spatial fluctuations which is typical for turbulent flows.

\begin{figure}[ht]
\begin{center}
\includegraphics[height=12.5cm]{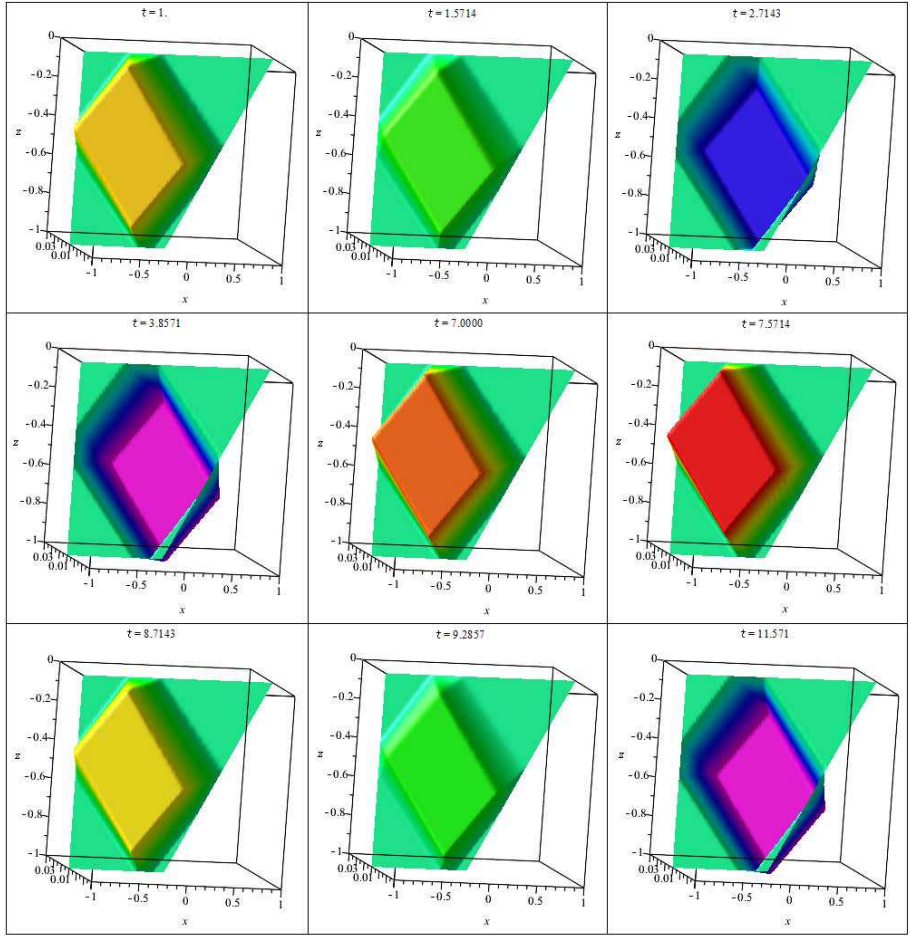}
\caption{The unconstrained 3d-plot of the stream function $\psi (x,z,t;0.65,\,0.7)$ at
some relatively small values of $t$.
}
\label{fig:ma8}
\end{center}
\end{figure}

\begin{figure}[ht]
\begin{center}
\includegraphics[height=12.5cm]{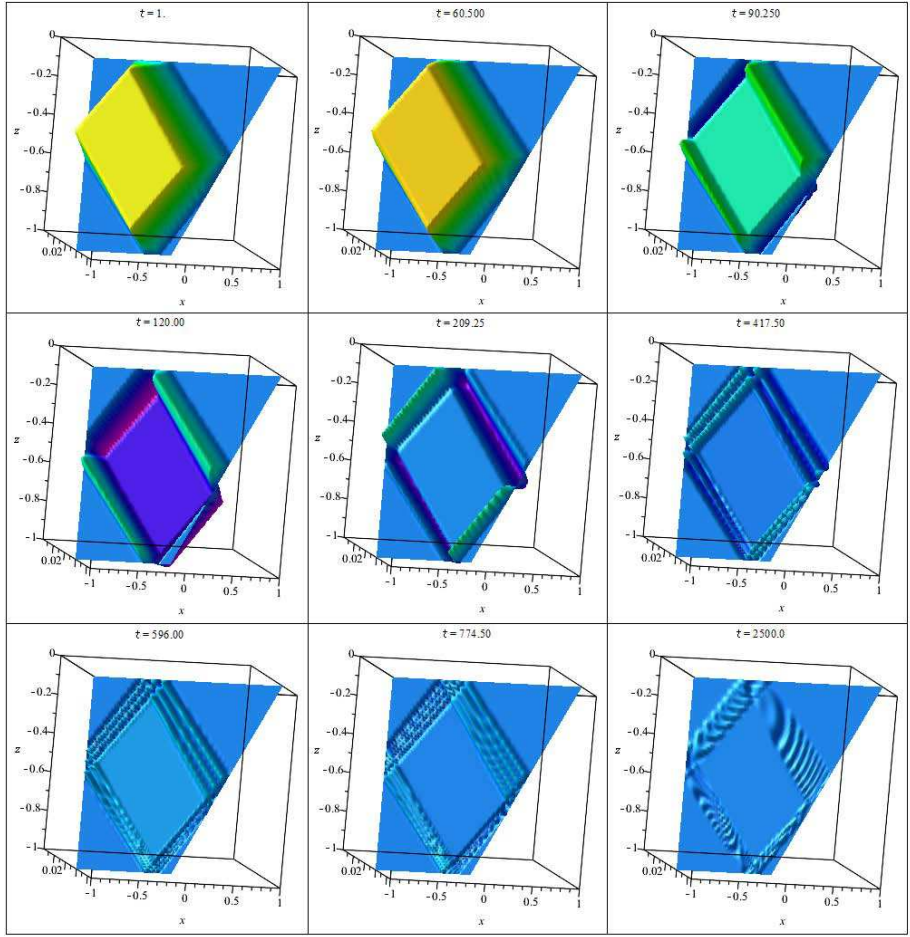}
\caption{
The unconstrained 3d-plot of the stream function $\psi (x,z,t;0.65,\,0.7)$ at some time moments.
}
\label{fig:ma8}
\end{center}
\end{figure}

\begin{figure}[ht]
\begin{center}
\includegraphics[height=13cm]{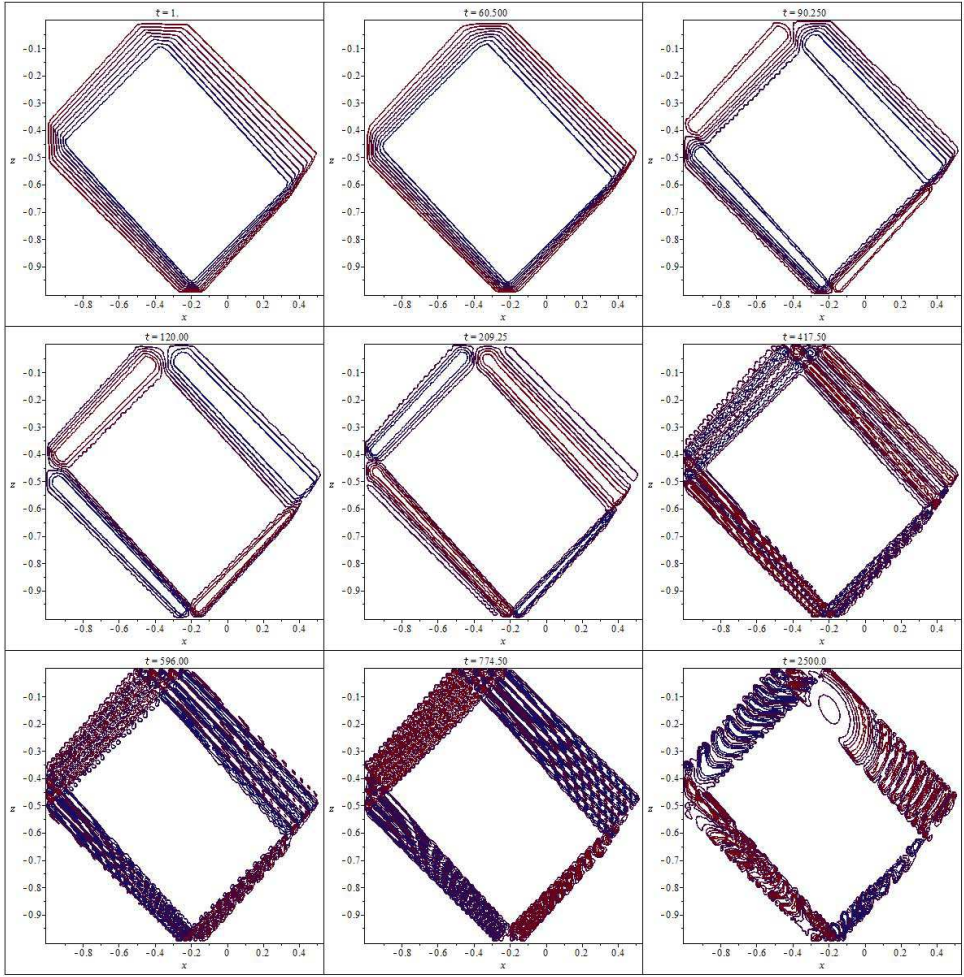}
\caption{The contour-plot of the stream function $\psi (x,z,t;0.65,\,0.7)$ at
some time moments.
}
\label{fig:ma8}
\end{center}
\end{figure}

\begin{figure}[ht]
\begin{center}
\includegraphics[height=13cm]{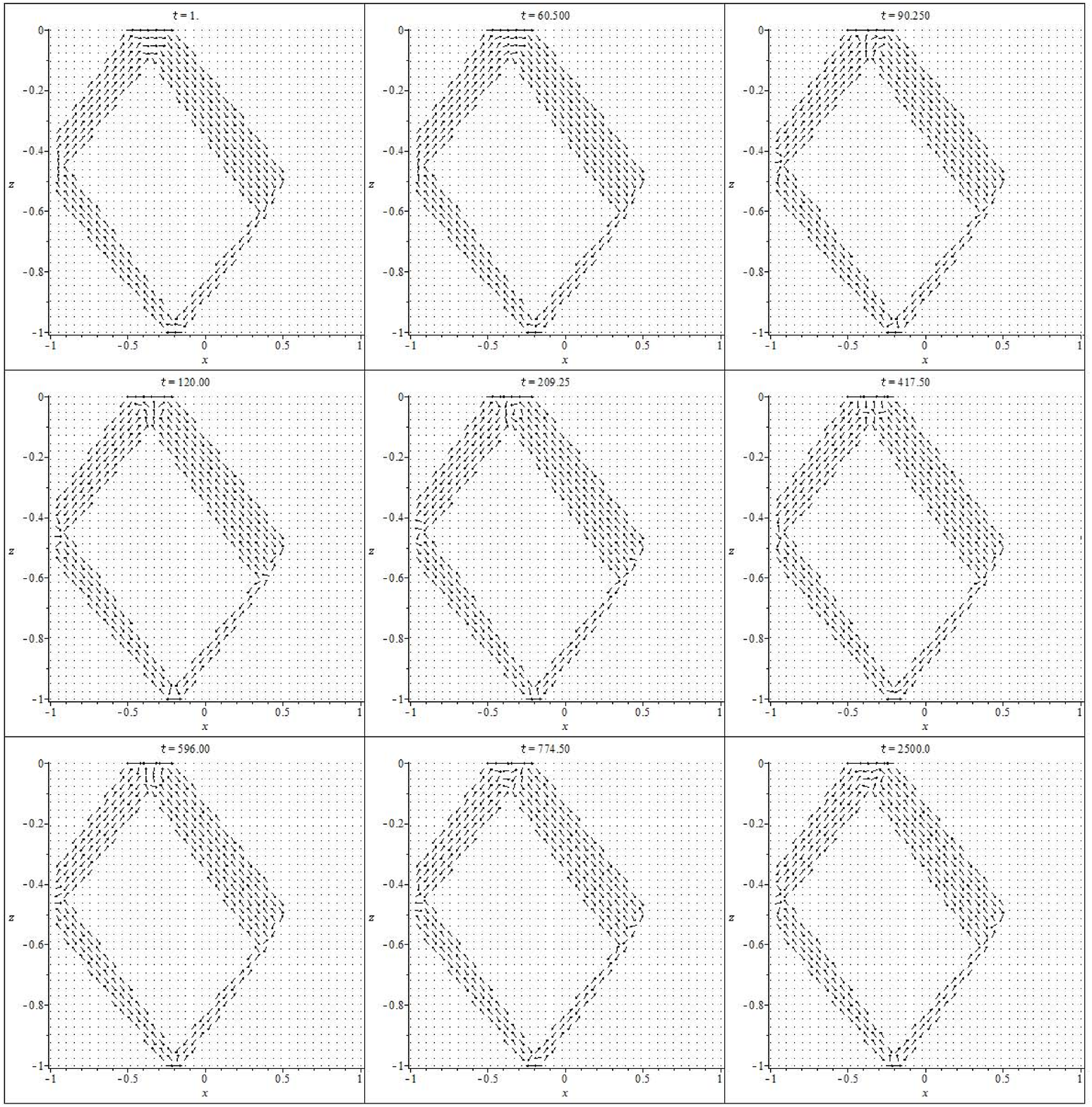}
\caption{The field-plot $[u(x,z,t), w(x,z,t)]$ at the different values of the variable $t$.}
\label{fig:ma8}
\end{center}
\end{figure}

\begin{figure}[ht]
\begin{center}
\includegraphics[height=13cm]{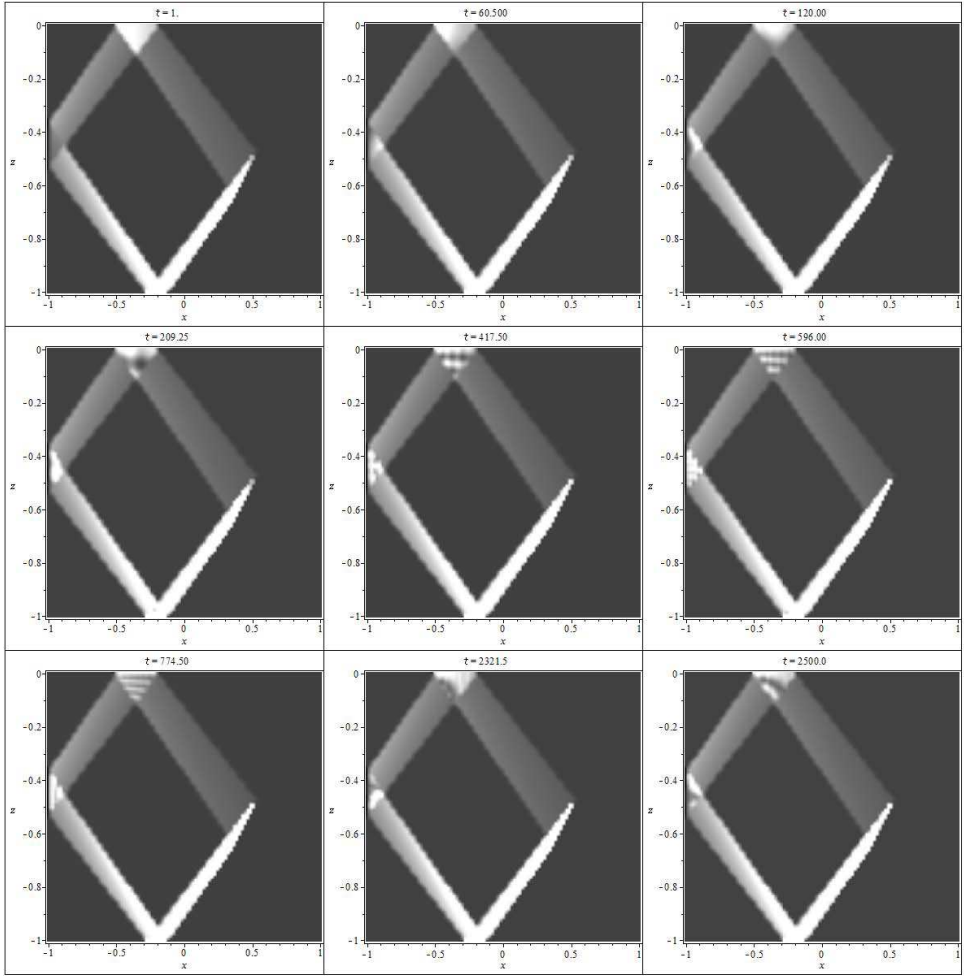}
\caption{
The plot of the energy density of the velocity field $\mathbf{U}$ at some time moments.
Note, please, that this function changes as time goes on only inside 4
small triangles and in the rest of the domain $D$ it remains constant.
}
\label{fig:ma8}
\end{center}
\end{figure}

\end{appendix}

\providecommand{\bysame}{\leavevmode\hbox to3em{\hrulefill}\thinspace}
\providecommand{\MR}{\relax\ifhmode\unskip\space\fi MR }
\providecommand{\MRhref}[2]{%
  \href{http://www.ams.org/mathscinet-getitem?mr=#1}{#2}
}
\providecommand{\href}[2]{#2}

%\bibliographystyle{amsplain}% BST file
%\bibliography{inwaves}% Bibliography database

\end{document}